\documentclass[parskip=full,12pt]{article}

\usepackage{amssymb}
\usepackage{amsthm}
\usepackage{amsmath}
\usepackage{booktabs}
\usepackage{graphicx}
\usepackage{natbib}
\usepackage{url}
\usepackage{color,soul}
\usepackage{multirow}
\usepackage[top=1in, bottom=1in, left=1in, right=1in]{geometry}
\usepackage[linesnumbered,ruled]{algorithm2e}
\usepackage{tikz}
\usetikzlibrary{er,positioning,arrows.meta,calc}
\usepackage[labelformat=simple]{subcaption}
\usepackage{fancyhdr}
\usepackage{fancyvrb}

\newcommand{\E}{\mathbb{E}}

\usetikzlibrary{er,positioning,arrows.meta,arrows,shapes,calc}

\tikzstyle{block} = [rectangle, draw, fill=white, 
    text width=7cm, text centered, rounded corners, minimum height=3em]

\tikzset{
     arrow/.style = { thick,  ->, >=Triangle},
}

\newtheorem{prop}{Proposition}
\def\spacingset#1{\renewcommand{\baselinestretch}%
{#1}\small\normalsize} \spacingset{1}

\title{Skewed link regression models for imbalanced binary response with applications to life insurance}

\author{Shuang Yin \thanks{Department of Statistics, University of Connecticut, 215 Glenbrook Road, Storrs, CT,  06269-4120, USA. Email: \texttt{shuang.yin@uconn.edu}.} 
\and Dipak K. Dey \thanks{Department of Statistics, University of Connecticut, 215 Glenbrook Road, Storrs, CT,  06269-4120, USA. Email: \texttt{dipak.dey@uconn.edu}.}\and Emiliano A. Valdez\thanks{Corresponding author: Department of Mathematics, University of Connecticut, 341 Mansfield Road, Storrs, CT, 06269-1009, USA. Email: \texttt{emiliano.valdez@uconn.edu}.}
\and Guojun Gan\thanks{Department of Mathematics, University of Connecticut, 341 Mansfield Road, Storrs, CT, 06269-1009, USA. Email: \texttt{guojun.gan@uconn.edu}.}
\and Jeyaraj Vadiveloo\thanks{Department of Mathematics, University of Connecticut, 341 Mansfield Road, Storrs, CT, 06269-1009, USA. Email: \texttt{jeyaraj.vadiveloo@uconn.edu}.}}

\begin{document}

\clearpage\maketitle
\thispagestyle{empty}

\begin{abstract}

For a portfolio of life insurance policies observed for a stated period of time, e.g., one year, mortality is typically a rare event. When we examine the outcome of dying or not from such portfolios, we have an imbalanced binary response. The popular logistic and probit regression models can be inappropriate for imbalanced binary response as model estimates may be biased, and if not addressed properly, it can lead to serious adverse predictions. In this paper, we propose the use of skewed link regression models (Generalized Extreme Value, Weibull, and Fre\`chet link models) as more superior models to handle imbalanced binary response. We adopt a fully Bayesian approach for the generalized linear models (GLMs) under the proposed link functions to help better explain the high skewness. To calibrate our proposed Bayesian models, we use a real dataset of death claims experience drawn from a life insurance company's portfolio. Bayesian estimates of parameters were obtained using the Metropolis-Hastings algorithm and for Bayesian model selection and comparison, the Deviance Information Criterion (DIC) statistic has been used. For our mortality dataset, we find that these skewed link models are more superior than the widely used binary models with standard link functions. We evaluate the predictive power of the different underlying models by measuring and comparing aggregated death counts and death benefits.

\vspace{0.6cm}


\end{abstract}

\newpage

\section{Introduction} \label{sec:intro}
\markboth{Introduction}{}

The prediction and interpretation of mortality rates have been an important aspect of risk management in the life insurance industry. For several decades, different models have been applied to actuarial work to investigate the possibilities of explaining mortality. The early work of \cite{gompertz1825} provided a parametric form that explains the exponential increase in mortality with age. More recent work includes \cite{haberman1996glm} and \cite{zhu2001cox}, which examined the applications of Cox proportional hazards model (\cite{cox1972regression}) in analyzing the mortality for life insurance policies. The Cox regression model is a tool for understanding the effects of the presence of risk factors such as age, gender, policy duration, and smoking habits, in mortality. \cite{olbricht2012tree} offered an interesting alternative approach using tree-based models to construct the life table. Accordingly, tree-based models have the advantages of ``ease of predictability'' and ``transparency.'' \cite{meyricke2013determinants} proposed the generalized linear mixed models (GLMMs) framework for mortality modeling to incorporate longitudinal variables allowing for both underwriting factors and frailty. \cite{booth2008review} discussed an array of work in the literature on mortality modeling especially those pertaining to forecasting.

This paper has a different objective for studying mortality. In particular, we are interested in tracking and monitoring the death claims experience of a portfolio of life insurance policies. As pointed out in \cite{yin2020cluster}, such tracking mechanism can help insurers take necessary actions to mitigate against the immediate economic impact of any deviations from expectations. An interesting related work of \cite{zhu2015logit} applies logistic regression models to understand mortality trends, slopes, and differentials. An additional challenge of mortality modeling, especially related to periodic claims monitoring, is that it is considered a rare event with a very tiny probability; there is a significant difference between the number of deaths and the number of survivors. For example, the mortality rate, according to our insured portfolio dataset, observed during a policy year is only about 1.3\%, which in a binary model setup, such a response is characterized as highly imbalanced (or skewed). In this paper, we focus on capturing this high skewness for which other methods may not be able to model the data appropriately.

It is well known that binary regression models fall within the class of generalized linear models (GLMs) using Bernoulli distribution. See \cite{mccullagh1989}. The commonly used functions linking the mean parameter to the predictor variables are the symmetric link functions, logit and probit, leading, respectively, to the logistic and probit regression models. Both link functions reasonably work for observations with balanced response binary outcomes. Suppose we are given a set of observations expressed as $(y_i, \boldsymbol{x}_i)$, where $y_i$ is a binary outcome and $\boldsymbol{x}_i$ is a set of predictor variables. The binary regression model has a latent variable representation where $y_i$, for observation $i$, is related to an unobserved variable $z_i$ as $y_i = I(z_i>0)$. $z_i$, also called the latent variable, is directly linked to the predictor variables as a linear model with an error component as $z_i = \boldsymbol{x}_i \boldsymbol{\beta} + u_i$, where $\boldsymbol{\beta}$ is a vector of coefficients for the predictors. The error component, $u_i$ has a pre-specified distribution directly related to the link function.

To deal with the limitations of symmetric link functions for binary outcomes, some studies extended the parametric class including the use of power transformation of the logit link (\cite{aranda1981two}) and the Box-Cox transformation for binary regression (\cite{guerrero1982boxcox}). Other model extensions include two-parameter classes of link functions such as two parameter generalization of logit link by \cite{pregibon1980goodness}, the two shape parameters link by \cite{stukel1988genlogistic}, and two tailed link by \cite{czado1994parametric}. However, the use of these extended classes of parametric link functions do not always fit the data well with extremely imbalanced response variable. For example, Stukel's generalized logistic model may lead to improper posterior distribution under many types of noninformative improper prior. \cite{chen1999new} suggested an asymmetric link function especially when the probability of a given binary response approaches 0 at a significantly different rate than it approaches 1. This model yields to proper posterior distribution under noninformative improper priors for the regression coefficients and can handle datasets with balanced and imbalanced response variable. However, the intercept and skewness parameters are confounded with each other in this model, although \cite{kim2007flext} overcame this problem to some extent with the constraint that the skewness lies in the range of $(0, 1]$, which restricts the possibility of negative skewness of the response variable. 

This open problem to some extent motivates us to investigate more flexible models to accommodate such imbalances. In order to appropriately model the large skewness, \cite{wang2010gev} proposed binary regression with the generalized extreme value (GEV) link. The GEV distribution belongs to a family of continuous probability distributions developed from extreme value theory. The GEV distribution combines the Gumbel, Fr\'echet, and reversed Weibull distributions into a single family. \cite{caron2018entropy} presented a Weibull link for the imbalanced binary regression, which is very flexible and capable to handle with different types of data. The Fre\'chet distribution, which is known as the Type II extreme value distribution, also has the shape parameter to control the direction and degree of skewness.

In this paper, we propose the use of this family of asymmetric link functions to binary regression models for analyzing life insurance data with highly imbalanced mortality. The model parameters are estimated using Bayesian approach. There are advantages to the use of these link functions that are suitable for our purposes:
\begin{itemize}
\item With a continuous range of possible values of the shape parameter, the resulting models can automatically estimate the skewness parameter without any pre-assumption on the direction of the skewness and can thus better quantify the skewness; 
\item The commonly used link functions, such as logit, probit, and complementary log-log, can be derived as special or limiting cases of this family of link functions;
\item In a Bayesian framework, even under the improper priors, the posterior distribution of the parameter is still proper; and
\item With the introduction of latent variables, the process can be computationally efficient whereby the Markov Chain Monte Carlo (MCMC) implementation can easily be used to sample the parameters.
\end{itemize}

This paper has been structured as follows.  Section \ref{sec:binary} provides a brief review of binary regression models and describes the latent variable interpretation of the binary regression models that is important for comparative purposes. We also describe and review the various skewed link models that are used in this paper. In Section \ref{sec:bayes}, we present ideas of Bayesian inference useful for the estimation and prediction based on our proposed skewed link regression models. For purposes of being self-contained, a detailed procedure of Metropolis-Hastings algorithm is described in this section. Section \ref{sec:simulation} discusses the simulation studies and shows the recovery work of the various proposed models. Section \ref{sec:empirical} is about the empirical data used in our model calibration and predictions. We conclude in Section \ref{sec:conclude}. 

\section{Binary regression models} \label{sec:binary}
\markboth{Binary regression models}{}

For a fixed period, our dataset can be described as $D = (\boldsymbol{y},\boldsymbol{x},n)$ where $\boldsymbol{y}=(y_1,y_2,\ldots,y_n)'$ is a vector of responses, $\boldsymbol{x}=(\boldsymbol{x}_1,\boldsymbol{x}_2,\ldots,\boldsymbol{x}_n)'$ is a matrix of predictor variables, and $n$ is the total number of observations. The response $y_i$ is a binary outcome, with $y_i=1$ indicating death has been observed for policyholder $i$ and $y_i=0$ indicating survival otherwise. The vector of predictor variables is $\boldsymbol{x}_i = (x_{i1},x_{i2},\ldots,x_{ip})'$, where $p$ is the number of observed predictor variables in the dataset.

Let $q_i = \text{Prob}(y_i=1|\boldsymbol{x}_i)$ denote the probability of death, given the predictor variables. Since $y_i$ is a binary outcome, it can easily be deduced that the conditional expectation of $y_i$ given $\boldsymbol{x}_i$ is also $\E(y_i=1|\boldsymbol{x}_i) = q_i = \text{Prob}(y_i=1|\boldsymbol{x}_i)$. It is well known that binary regression model falls within the class of generalized linear models (GLMs) with the link function $g(q_i) = \boldsymbol{x}'_i \boldsymbol{\beta}$ where $\boldsymbol{\beta}=(\beta_0,\beta_1,\beta_2,\ldots,\beta_p)'$ is a vector of coefficients. Therefore, $q_i = g^{-1}(\boldsymbol{x}'_i \boldsymbol{\beta})$ and we would conveniently write this as $q(\boldsymbol{x}'_i \boldsymbol{\beta})$. Note that it is customary to include an intercept coefficient $\beta_0$ in the linear relationship so that without loss of generality, we can assume that the vector of predictor variables is augmented as $(1,\boldsymbol{x}_{i1},\ldots,\boldsymbol{x}_{ip})'$. See \cite{mccullagh1989} and \cite{cox1989}.

\subsection{Latent variable interpretation} \label{sec:latent}

The binary regression model can be re-expressed and re-interpreted in terms of latent (or unobserved) variables, herewith denoted by $\boldsymbol{z}=(z_1,z_2,\ldots,z_n)'$.  Define these latent variables so that the binary outcome is
\begin{equation}
y_i = \begin{cases}
1, & \text{for } z_i > 0 \\
0, & \text{for } z_i \le 0 \\
\end{cases}
= I(z_i > 0), \label{eq:latent1}
\end{equation} 
where
\begin{equation}
z_i = \boldsymbol{x}'_i \boldsymbol{\beta} + u_i, \label{eq:latent2}
\end{equation}
and $u_i|\boldsymbol{x}'_i \sim F$. Here $F$ is a cumulative distribution function of $u_i$, given $\boldsymbol{x}'_i$, which is often assumed to be a continuous random variable. It follows therefore that
\begin{equation}
\begin{split}
q(\boldsymbol{x}'_i \boldsymbol{\beta})  &= \E(y_i=1|\boldsymbol{x}_i) = \text{Prob}(z_i >0)  \\
        &= \text{Prob}(u_i > - \boldsymbol{x}'_i \boldsymbol{\beta}) = 1-F(-\boldsymbol{x}'_i \boldsymbol{\beta}). \label{eq:latent3}
\end{split}
\end{equation}
Note that in the case where $F$ is the distribution function of a symmetric random variable $u_i$ with mean 0, then equation (\ref{eq:latent3}) becomes
\begin{equation*}
q(\boldsymbol{x}'_i \boldsymbol{\beta}) = F(\boldsymbol{x}'_i \boldsymbol{\beta}).
\end{equation*}
In this case, $F^{-1}$ determines the link function in the GLM framework. This latent variable representation is believed to be introduced by \cite{albert1993bayes}.

Examples of distributions that belong to the class of symmetric distributions include the standard logistic, standard normal, and the standard Gumbel distributions. This leads respectively to the three commonly used link functions in binary regression (without loss of generality, the subscript $i$ is dropped):
\begin{itemize}
\item[(1)] \textit{Logistic regression}: When $u$ has a logistic distribution with mean 0 and scale parameter 1, i.e., $u \sim \text{logistic}(0,1)$, we have $F(u) = 1/(1+e^{-u})$. The link function can be expressed as the logit function $g(q) = \log(q/(1-q))$.
\item[(2)] \textit{Probit regression}: When $u$ has a normal distribution with mean 0 and scale parameter 1, i.e., $u \sim \text{Normal}(0,1)$, we have $F(u) = \Phi(u)$. The link function can be expressed as the probit function $g(q) = \Phi^{-1}(q)$.
\item[(3)] \textit{cloglog regression}: When $u$ has a Gumbel distribution with location 0 and scale parameter 1, we have $F(u) = \exp(-\exp(-u))$. The link function can be expressed as the cloglog function $g(q) = \log(-\log(1-q))$.
\end{itemize}

There are several advantages to the latent variable representation of the binomial regression. First, it allows for a natural interpretation. To illustrate, in the mortality context where $y=1$ is death and $y=0$ is survival, then the latent variable may be viewed as a continuous measure of mortality. Second, the latent variable allows for a natural cutpoint, at zero, for splitting the prediction for $y$ into one of a binary outcome. Third, it has computational advantages particularly when Bayesian inference is used for estimation such as the use of an MCMC algorithm that may include Metropolis-Hastings and Gibbs sampling. See \cite{gelman2014}. While a direct likelihood estimation is sufficient for simple models such as logistic and probit regression, Bayesian computation is particularly helpful for other forms of link functions such as those asymmetric link functions introduced in the subsequent subsections.

\subsection{Skewed link functions} \label{sec:asym}

The use of symmetric link functions in binary regression may be considered too restrictive especially for observations with imbalanced response, one where there is considerably a large proportion of one outcome.  In the case of imbalanced data, the link function can be mis-specified, which can lead to inconsistent and biased model estimates and therefore increase in the values of mean squared errors. See \cite{czado1994parametric}. This paper focuses on using parametric link functions based on distributions that are not necessarily symmetric. In particular, we examine and adopt skewed link functions derived from three families of distributions that can accommodate large skewness to handle the imbalanced response. They are the generalized extreme value (GEV) distributions, the skewed Weibull distributions, and the Fre\'chet distributions.

For our purposes, we will consider three families of asymmetric or skewed link functions that are based on the generalized extreme value (GEV) distributions. Its cumulative distribution function has the general form
\begin{equation}
G(u;\mu,\sigma,\xi) = \exp \left\{-\left[1+ \xi \left(\frac{u-\mu}{\sigma} \right) \right]^{-1/\xi} \right\}, \label{eq:gev}
\end{equation} 
where $-\infty < u < \infty$ provided it satisfies $1+ \xi (u-\mu)/\sigma >0$, and the parameters belong to $-\infty < \mu < \infty$, $\sigma >0$, and  $-\infty < \xi < \infty$. This GEV distribution function can be derived from a limiting approximation to the distribution of the maxima of a sequence of random variables. This form has been attributed to the original work of \cite{mcfadden1978gev}. 

After some preliminary investigation of suitable link functions for our purposes, we have concluded to consider the following three families of skewed link functions:
\begin{itemize}
\item[(1)] \textit{Standard GEV link functions}: When the location is $\mu=0$ and the scale $\sigma=1$, we will call this the standard GEV distribution so that its distribution function in this case is
\begin{equation}
F_{\text{SGEV}}(u) = \exp [-(1+ \xi u)^{-1/\xi}]. \label{eq:SGEV}
\end{equation}
We will not preclude the case where $\xi=0$, in which case the distribution has the simplified form of $F(u) = \exp [- \exp(-u)]$. This is the case of the Gumbel link functions, based on the Gumbel or Type I extreme value distribution. The link function can be expressed as $g(q) = [ (-\log q)^{-\xi} -1]/ \xi$.
\item[(2)] \textit{Skewed Weibull link functions}: In this case, we have 
\begin{equation}
F_{\text{SW}}(u) = 1- \exp(-u^\gamma) \label{eq:SW}
\end{equation}
defined only for $u>0$; its value will be zero elsewhere. This is actually the standard form of the Weibull distribution when the location is $\mu=0$ and the scale $\sigma=1$, which is indirectly implied from the GEV form in (\ref{eq:gev}) above. The reverse Weibull, also called the Type III extreme value distribution can be derived directly from this GEV distribution; the derivation of the Weibull considers the minima of a sequence of random variables, in contrast to maxima as in the GEV. The link function can be expressed as $g(q) = [-\log(1-q)]^{1/\gamma}$.
\item[(3)] \textit{Fre\'chet link functions}: When the location is $\mu=0$, the scale $\sigma=1$, $\alpha = 1/\xi$, and some transformation to $1+\xi u$, we arrive at the Fre\'chet distribution function of the form
\begin{equation}
F_{\text{FR}}(u) = \exp(-u^{-\alpha}), \label{eq:FR}
\end{equation}
provided $u>0$; it is defined to be 0 elsewhere. This distribution is sometimes referred to as the Type II extreme value distribution. The link function can be expressed as $g(q) = (-\log q)^{-1/\alpha}$.
\end{itemize}

Briefly, let us examine some properties of each of these distributions. For the standard GEV, the skewness measure proposed by \cite{arnold1995skew} is given by $\gamma=1-2F_{\text{SGEV}}(M)$, where $M$ is the mode of the distribution. Based on this definition, the flexibility of the link function is reflected in the skewness expressed explicitly as $2\exp(-(1+\xi))-1$, which depends only on the shape parameter $\xi$. The mode is evaluated at $[(1+\xi)^{-\xi}-1]/\xi$. 

For the skewed Weibull distribution, the shape parameter $\gamma \in (0, +\infty)$ controls the skewness. Similar to GEV link, the skewness can also be computed based on the mode $(\frac{\gamma-1}{\gamma})^{\frac{1}{\gamma}}$, which relies only on the shape parameter $\gamma$. Based on the formula $1-2\exp(-(\frac{\gamma-1}{\gamma}))$, the skewness lies somewhere in the interval $(-\infty, 0.2642)$.

The complementary log-log (cloglog) model, which has been considered as an alternative to logistic and probit models, is frequently used when the probability of an event is relatively small or relatively large. The complementary log-log link is a special case of the Weibull link with corresponding cumulative distribution function $F(u) = 1-\exp(-\exp(u))$. Consider the general form of the Weibull distribution expressed as $F(u)= 1-\exp (-(u-\alpha)^{\gamma})$. If we set $\alpha=-1$ and divide $u$ by $\gamma$, then we have
\begin{equation}
\lim_{\gamma \rightarrow \infty} 1-\exp \left[- \left(1+\frac{u}{\gamma} \right)^{\gamma} \right] = 1-\exp[-\exp(u)].
\end{equation} 
Although considered a skewed link, the cloglog, with skewness defined as $1-2F(M)=1-2F(0)=-0.2642$, does not provide enough flexibility since the mode $M=0$ is a constant without incorporating a shape parameter.

Finally, for the Fre\'chet distribution, the skewness parameter can be estimated by the mode using the relationship $1-2F_{\text{FR}}(M)$, where the mode $M=[\alpha/(1+\alpha)]^{1/\alpha}$. Thus,the skewness of this distribution is $2\exp[-(1+\alpha)/\alpha]-1$, which relies only on the shape parameter $\alpha$ and lies in the interval $(-1, -0.2642)$. This explains why this can be considered a flexible link function. Another important feature of the shape parameter arises from the fact that this kind of parameter affects the shape of the distribution by purely controlling the tail behavior than simply shifting it or stretching it.

Figure \ref{fig1:densities} displays the probability densities of asymmetric distributions described in this subsection (standard GEV, Weibull, and Fre\'chet) under different shape parameters. This allows us to visualize the type of skewness that can be derived from each of these distributions.

\begin{figure}[htbp!]
    \centering
    \scalebox{0.5}{
    \includegraphics{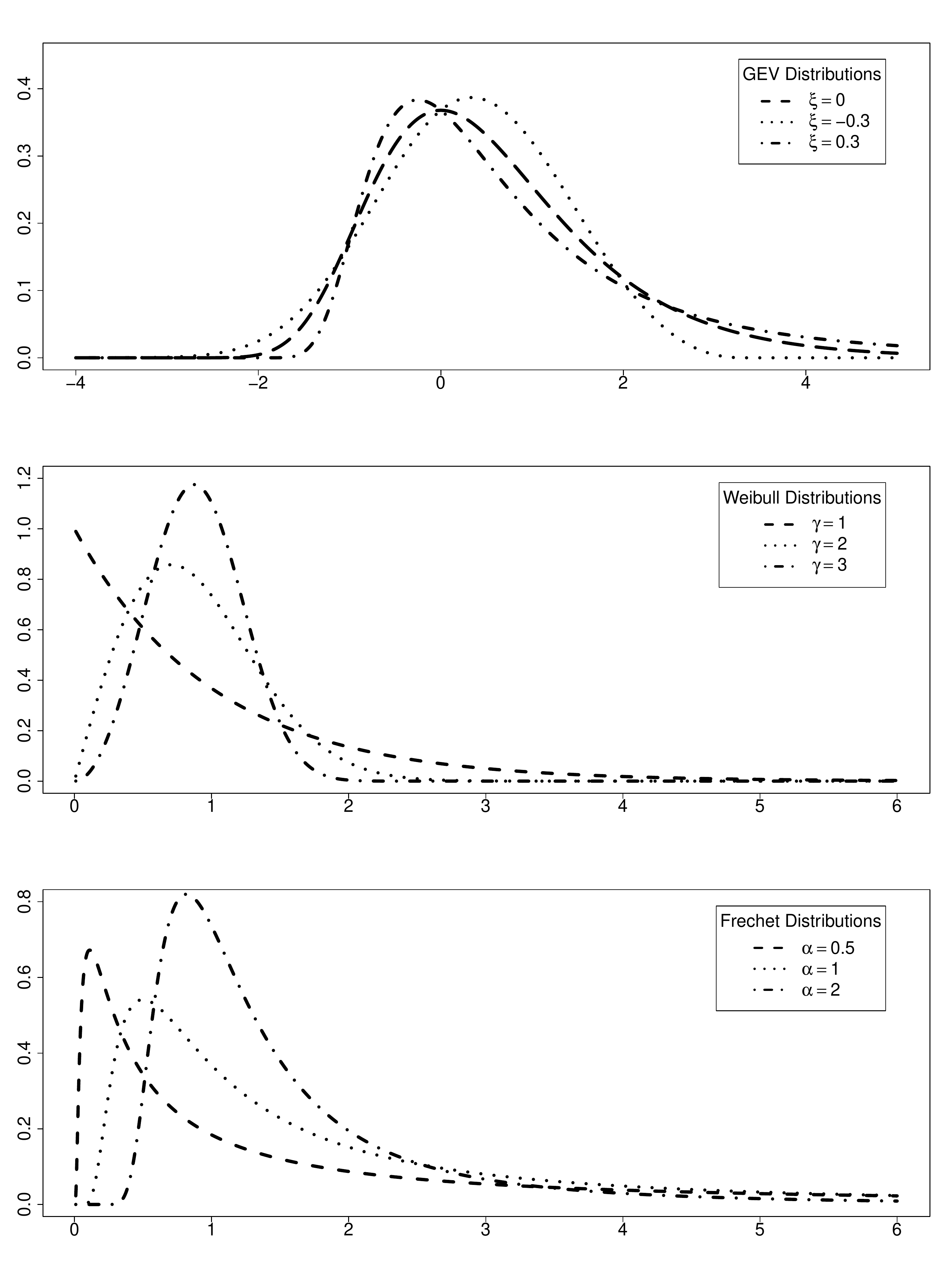}
    }
    \caption{Density plots for the standard GEV with $\xi=0,-0.3, 0.3$, skewed Weibull with $\gamma=1, 2, 3$, and Fre\'chet with $\alpha=0.5, 1, 2$}
    \label{fig1:densities}
\end{figure}

\section{Bayesian inference} \label{sec:bayes}
\markboth{Bayesian inference}{}

For the binary regression models with standard link functions that include the logit, probit, and cloglog, the parameters can be estimated based on the maximum likelihood using either the Newton's method or the iteratively reweighted least squares (IRLS). See \cite{mccullagh1989}. For our binary regression based on the skewed link models, we used the Bayesian approach with approximation of the posterior distribution based on Metropolis-Hastings algorithm, a type of a Markov chain Monte Carlo (MCMC) simulation method.

\subsection{Posterior distribution and inference}
In Bayesian statistics, the posterior distribution of the parameters is the conditional probability, given relevant evidence or information, which is usually the dataset. Let $\pi(\theta|y)$ denote the posterior distribution of the parameter $\theta$, given the dataset $y$, and $L(\theta|y) = p(y| \theta)$ denote the likelihood function of $\theta$, where $p(y|\theta)$ is the probability density function of $y$, given $\theta$. Using the Bayes Theorem, if we assign a prior distribution $\pi(\theta)$ to $\theta$, we have 
\begin{equation}
\pi(\theta|y) = \frac{p(y, \theta)}{\displaystyle \int_{\Theta} p(y|\theta)\pi(\theta) \text{d}\theta}= \frac{p(y|\theta)\pi(\theta)}{\displaystyle \int_{\Theta} p(y|\theta)\pi(\theta) \text{d}\theta} \propto  p(y|\theta)\pi(\theta) = L(\theta|y) \times \pi(\theta), \label{posterior}
\end{equation}
where $\displaystyle \int_{\Theta} p(y|\theta)\pi(\theta) \text{d}\theta$ is the normalized constant that is free of $\theta$. Suppose that we sample some number $N$ of independent, random values of $\theta$ from the posterior distribution $\pi(\theta|y)$, $\theta^{(1)}, \theta^{(2)}, \ldots, \theta^{(N)} \overset{i.i.d}{\sim} \pi(\theta|y)$. Then according to the  Law of Large Numbers,
\begin{equation}\label{posterior_expect}
    \begin{aligned}
    \E(\theta|y)& \approx \frac{1}{N} \sum_{i=1}^N \theta^{(i)} \\
    \E(g(\theta)|y) & \approx \frac{1}{N} \sum_{i=1}^N g(\theta^{(i)}),
    \end{aligned}
\end{equation}
where $g(\theta)$ is any function of $\theta$. As $N$ gets larger, the sample mean converges to its expectation. In this paper, the samples $\theta^{(i)}$'s are generated from the posterior distribution of equation (\ref{posterior}) using MCMC (Metropolis-Hastings) algorithm which is discussed in details in section (\ref{sec:mh}).

\subsection{Choice of priors for the skewed link models}

As earlier introduced, we denote our observed dataset as $D = (\boldsymbol{y},\boldsymbol{x},n)$, and for purposes of the subsequent discussions, we shall denote our set of parameters as $(\boldsymbol{\beta},\zeta)$ where without loss of generality, $\boldsymbol{\beta}$ corresponds to the regression coefficients for the predictor variables and $\zeta$ denotes any additional parameters arising from the skewed link models. Given the data $D$, the likelihood of the parameters can be written as:
\begin{equation}
L(\boldsymbol{\beta}, \zeta| D)  \propto \prod_{i=1}^n q_{i}^{y_i} (1-q_i)^{1-y_i} = \prod_{i=1}^n (1- F(-\boldsymbol{x_i'\beta}))^{y_i} \times
    F(-\boldsymbol{x_i'\beta})^{1-y_i}. \label{eq:likelihood}
\end{equation}
According to equation (\ref{posterior}), if the joint prior of $(\boldsymbol{\beta}, \zeta)$, $\pi(\boldsymbol{\beta}, \zeta)$, is given, then the joint posterior distribution for $(\boldsymbol{\beta}, \zeta)$ can be written as follows:
\begin{equation} \label{eq:joint}
\begin{aligned}
\pi(\boldsymbol{\beta}, \zeta|D) & \ \propto  L(\boldsymbol{\beta}, \zeta| D) \times \pi(\boldsymbol{\beta}, \zeta) \\
    & = \prod_{i=1}^n (1- F(-\boldsymbol{x_i'\beta}))^{y_i} \times F(-\boldsymbol{x_i'\beta})^{1-y_i} \times \pi(\boldsymbol{\beta}, \zeta).
\end{aligned}
\end{equation}

Consider the case of the \textit{standard GEV} link model where the additional parameter is $\zeta = \xi$. If the independent uniform priors $\pi(\boldsymbol{\beta})\propto 1$, $\pi(\xi) \propto 1$ are assigned to $\boldsymbol{\beta}$ and $\xi$, then $\pi(\boldsymbol{\beta}, \xi) \propto \pi(\boldsymbol{\beta}) \pi(\xi) \propto 1$. The resulting posterior distribution is still proper (\cite{wang2010gev}). Assuming independent priors for $\boldsymbol{\beta}$ and $\xi$, equation (\ref{eq:joint}) becomes
\begin{equation} \label{eq:jointSGEV}
\begin{aligned}
\pi(\boldsymbol{\beta}, \xi|D) \propto  \ \prod_{i=1}^n &  \left (1-\exp\left[-(1-\xi \boldsymbol{x_i'\beta})^{-1/\xi} \right] \right)^{y_i} \\
   & \times \left (\exp\left[-(1-\xi \boldsymbol{x_i'\beta})^{-1/\xi} \right] \right)^{1-y_i} \times \pi(\boldsymbol{\beta}) \times \pi(\xi).
\end{aligned}
\end{equation}
In this paper, we assume the independent multivariate normal and univariate normal priors for $\boldsymbol{\beta}$ and $\xi$, respectively, with $\boldsymbol{\beta}\sim \text{N}_{p+1} (\boldsymbol{0}, \sigma^2_{\beta} \boldsymbol{I}_{p+1})$ and $\xi\sim \text{N}(0, \sigma^2_{\xi})$.

Consider the case of the \textit{skewed Weibull} link model where the additional parameter is $\zeta = \gamma$. If the informative priors for $\boldsymbol{\beta}$ and $\gamma$ cannot be obtained, we can assign uniform prior $\pi(\boldsymbol{\beta}) \propto 1$ and noninformative prior $\pi(\gamma) \propto 1/\gamma^c$, for $\gamma >1$ and $c>1$, a fixed known constant. The resulting propriety of the posterior under improper priors, in this case,  is proved in \cite{caron2018entropy}. Assuming independent priors for $\boldsymbol{\beta}$ and $\gamma$, equation (\ref{eq:joint}) becomes
\begin{equation} \label{eq:jointSW}
\pi(\boldsymbol{\beta}, \xi|D) \propto  \ \prod_{i=1}^n \left(\exp\left [-(-\boldsymbol{x_i'\beta})^{\gamma} \right] \right)^{y_i} \times \left(1-\exp\left[-(-\boldsymbol{x_i'\beta})^{\gamma} \right]\right)^{1-y_i} \times \pi(\boldsymbol{\beta}) \times \pi(\gamma).
\end{equation}
In this paper, we assume $\boldsymbol{\beta}, \gamma$ are priori independent and we assign a multivariate normal N$_{p+1}(\boldsymbol{0}, \sigma^2_{\beta}\boldsymbol{I}_{p+1})$ and Gamma$(\text{shape} =3, \text{rate} = 4)$ to $\boldsymbol{\beta}$ and $\gamma$, respectively.

Finally, consider the case of the \textit{Fre\'chet} link model where the additional parameter is $\zeta = \alpha$. When the information of the prior is unavailable, we can consider the noninformative prior for $\boldsymbol{\beta}$ and $\alpha$. Thus, $\pi(\boldsymbol{\beta}, \alpha) \propto \pi(\beta)\pi(\alpha) \propto \frac{1}{\alpha^c}$, $\alpha > 0$ and $c>1$ is a constant. Given this noninformative prior, the posterior distribution explicitly written below as equation (\ref{eq:jointFR}) is proper. The resulting propriety of the underlying posterior distribution is proved in the Appendix. Such results have not appeared in the literature and therefore we find it useful to provide enough details of the proof. Assuming independent priors for $\boldsymbol{\beta}$ and $\gamma$, equation (\ref{eq:joint}) becomes
\begin{equation} \label{eq:jointFR}
\pi(\boldsymbol{\beta}, \xi|D) \propto  \ \prod_{i=1}^n \left(1-\exp\left[-(-\boldsymbol{x_i'\beta})^{-\alpha} \right]\right)^{y_i} \times \left(\exp\left[-(-\boldsymbol{x_i'\beta})^{-\alpha} \right]\right)^{1-y_i} \times \pi(\boldsymbol{\beta}) \times \pi(\alpha).
\end{equation}
In this paper, the multivariate normal and Gamma distributions are considered to be the priors of $\boldsymbol{\beta}$ and $\alpha$, respectively, with a similar parameterization of the Gamma distribution as in the case of the skewed Weibull link model.

\subsection{Metropolis-Hastings algorithm} \label{sec:mh}

The Metropolis-Hastings algorithm is used to sample the regression coefficients $\boldsymbol{\beta}$ and the shape parameters $\xi$, $\gamma$, and $\alpha$ for the standard GEV, skewed Weibull, and Fre\'chet link models, respectively. The Metropolis-Hastings algorithm (\cite{hastings1970}) can generate a random walk based on a proposed distribution and through a rejection method to determine the move. To fix ideas, let $p(\cdot|D)$ and $Q(\cdot)$ denote posterior and proposed distributions with given complete data $D$ and let $\theta$ denote the parameters. The procedure of Metropolis-Hastings algorithm can be described in steps as follows:
\begin{itemize}
    \item Step 0: Choose an arbitrary starting point $\theta_0$ and set $i=0$;
    \item Step 1: Generate a candidate point $\theta_{can}$ from the proposed distribution $Q(\theta_i, \cdot)$ and $u$ from Uniform(0, 1);
    \item Step 2: if $u\leq a(\theta_i, \theta_{can})$ then set $\theta_{i+1}=\theta_{can}$, otherwise, set $\theta_{i+1}=\theta_{i}$, the acceptance probability is given by: 
    \begin{align*}
        a(\theta_i, \theta_{can}) = \min \left \{1, \frac{p(\theta_{can}|D) / Q(\theta_i| \theta_{can})}{p(\theta_{i}|D) / Q(\theta_{can}|\theta_i)} \right \};
    \end{align*}
    \item Step 4: Update $i = i+1$ and repeat steps 1 and 2 until convergence. 
\end{itemize}

The distribution $Q$ is the kernel distribution. If $Q$ is symmetric, then the conditional probability of $Q(\theta_i| \theta_{can})$ and $Q(\theta_{can}|\theta_i)$ can be cancelled out in the calculation of $a(\theta_i, \theta_{can})$. In this paper, we perform 20,000 iterations for the parameter estimation after a burn-in of 1,000 iterations. Let $\theta^{(i)}$ for $i=1, 2, \ldots, N$ denote the samples which are draw from Metropolis-Hastings, and by the Law of Large Number, the posterior mean of $p(\theta)$ is
\begin{equation}
    \bar{p}(\theta) = \frac{1}{N} \sum_{i=1}^N p(\theta^{(i)}) \rightarrow \E[p(\theta)|D]
\end{equation}
with variance estimate
\begin{equation}
    \hat{\sigma}^2 = \frac{1}{N^2} \sum_{i=1}^N [p(\theta^{(i)})-\bar{p}(\theta))]^2 \rightarrow \text{Var}[\bar{p}(\theta)]. 
\end{equation}
We choose normal kernels N$(\text{mean} = \theta^{(i-1)}, \text{variance})$ for the Metropolis-Hastings sampling, which is a normal distribution with mean of the previous state and variance with an appropriate value. The subsampling is also used in this paper, in which the samples are taken from every 50 steps in order to reduce the strong autocorrelation between states in order to further guarantee the stability of the estimates. 

\subsection{Model assessment and evaluation}

For the model fitting in a Bayesian routine, we use the Deviance Information Criterion (DIC) which was proposed by \cite{spiegelhalter2002bayesian} to perform the model comparison and assessment.  

The Deviance is defined as $-2 \times \log(f(\boldsymbol{y}|\theta)$, then
\begin{equation}
    \text{DIC} =2 \times \overline{\text{Deviance}}(\boldsymbol{y}, \theta) - \text{Deviance}(\boldsymbol{y}, \overline{\theta}),
\end{equation}
where $\text{Deviance}(\boldsymbol{y}, \overline{\theta})$ is the Deviance evaluated under the value of the posterior mean $\overline{\theta}$ of the corresponding parameters, $\overline{\text{Deviance}}(\boldsymbol{y}, \theta) = \frac{1}{N} \sum_{i=1}^N \text{Deviance}(\boldsymbol{y}, \theta^{(i)})$ which is the average estimated discrepancy for $N$ samples, and $\theta^{(i)}$ is the $i$th sample generated from the posterior distribution $\pi(\theta|D)$. We also use $p_D=\overline{\text{Deviance}}(\boldsymbol{y}, \theta) - \text{Deviance}(\boldsymbol{y}, \overline{\theta})$ to measure the effective model size and penalize complexity. The DIC is a Bayesian alternative to AIC and BIC. The model with smaller DIC will provide better fit to the data. 

For the comparison of the frequentist models, we use the Bayesian Information Criterion (BIC), as developed in \cite{schwarz1978bic}. Recall that $\text{BIC} = \text{Deviance}(\hat{\theta})+p\times \log(n)$, where $\hat{\theta}$ is the maximum likelihood estimates and $p$ is the model dimension. We use Kolmogorov-Smirnov statistics (KS) as a measure of goodness of fit, which is defined as $\max_i|y_i-\hat{y}_i|$, the maximum absolute error of the predicted and observed values. Also, another statistics MAE, Mean Absolute Error, defined as $\frac{1}{n} \sum_{i=1}^n |y_i-\hat{y}_i|$, is used to measure the absolute deviations.

\subsection{Posterior predictive distribution} \label{predictive}

The posterior predictive distribution is the distribution of possible unobserved values conditional on the observed values. For a given dataset $\boldsymbol{y}=(y_1, y_2, \ldots, y_n)$, the posterior predictive distribution of a new unobserved value $y_{n+1}$ is defined to be the following
\begin{equation}
f(y_{n+1}|\boldsymbol{y})  = \int f(y_{n+1}, \boldsymbol{\theta}|\boldsymbol{y}) d \boldsymbol{\theta} = \int f(y_{n+1}| \boldsymbol{\theta}, \boldsymbol{y}) \pi(\boldsymbol{\theta}| \boldsymbol{y}) d \boldsymbol{\theta}. \label{eq:postpredic}
\end{equation}

If we assume the observed and unobserved data are conditionally independent given $\boldsymbol{\theta}$, then equation (\ref{eq:postpredic}) can be written as
\begin{equation}
f(y_{n+1}|\boldsymbol{y})  = \int f(y_{n+1}|\boldsymbol{\theta}) \pi(\boldsymbol{\theta}| \boldsymbol{y}) d \boldsymbol{\theta}.
\end{equation}
MCMC samples from the posterior predictive distribution of $\boldsymbol{Y}$ can be obtained following the procedure described below (\cite{hoff2009}). For each $s \in \{1, 2, \ldots, S\}$, 
\begin{itemize}
\item sample $\theta^{(s)} \sim \pi(\boldsymbol{\theta}| \boldsymbol{Y} = \boldsymbol{y}_{observed})$; and 
\item sample $\tilde{\boldsymbol{Y}}^{(s)} = (\tilde{y}_1^{(s)}, \ldots, \tilde{y}_n^{(s)}) \sim \text{i.i.d} \ \ f(y|\theta^{(s)})$ .
 \end{itemize}
It can be inferred that the sequence $\{ (\theta, \tilde{\boldsymbol{Y}})^{(1)}, \ldots, (\theta, \tilde{\boldsymbol{Y}})^{(S)} \}$ constitutes $S$ independent samples from the joint posterior distribution of $(\theta, \tilde{\boldsymbol{Y}})$ and the sequence $(\tilde{\boldsymbol{Y}}^{(1)}, \ldots, \tilde{\boldsymbol{Y}}^{(S)})$ also constitutes $S$ independent samples generated from the posterior predictive distribution of $\tilde{\boldsymbol{Y}}$. 

\section{Simulation studies and parameter recovery} \label{sec:simulation}
\markboth{Simulation studies and parameter recovery}{}

In the simulation studies, we conducted two experimental trials and to investigate and assess the performance of the various link models for each of these experiments. For each of the simulated dataset, we deploy six models under different link functions, logit, probit, and cloglog, which are considered frequentist models, together with skewed link models, standard GEV, skewed Weibull, and Fr\'echet, which are considered Bayesian models. To summarize the posterior marginal densities of the parameters under a Bayesian framework, we use the highest posterior density interval (\cite{box1973}), which has the shortest length for a given probability content. The calculations in this paper use the \cite{chen1999monte} algorithm to estimate an empirical Highest Posterior Density (HPD) interval of the parameters.

For experiment I, we generate 1000 independent response variables from the cloglog regression. The predictor $x_1$ is a categorical variable with two levels (with respective probabilities $0.3$ and $0.7$) and $x_2$ is a numerical attribute generated from N$(0,1)$. The true values of the coefficients $(\beta_0, \beta_1, \beta_2)$ are set to be $(-3, 0.2, 0.7)$ and the linear component of the regression model is $\boldsymbol{x'\beta} = \beta_0 + x_1\beta_1 +x_2 \beta_2$. The dataset contains responses of 15 1's and 985 0's. The estimation results based on the various link models are summarized in Table \ref{tab:sim_probit}, which shows that the skewed Weibull link model performs better than other Bayesian models with the lowest DIC and the largest log likelihood. The KS statistics for all the links are approximately close to each other, however, the Fre\'chet model outperforms all other models with respect to MAE. When it comes to parameter recovery, the asymmetric link models do not necessarily outperform the true probit model but are competitive enough against the other symmetric links. For the logit and cloglog models, the coefficient estimates of the significantly deviate from the true values; for all the skewed link models, most of the estimates fall within the 95\% HPD interval.

For experiment II, we generate the response variables with sample size $N=1000$ from the GEV regression model with $\xi = -0.3$ and $(\beta_0, \beta_1, \beta_2) = (-3, -0.2, 0.8)$, and that produced responses of 12 1's and 988 0's. The predictors are set to be the same with the previous experiment. The output, summarized in Table \ref{tab:sim_gev}, presents the model estimates under the six different link models. The data modeling of the standard GEV model performs well not only in view of recovery of parameters, but also in the model assessment. However, the skewed Weibull link is the most appropriate one since in both of the simulations, its estimated parameters are the ones that closely reach the true parameter values. The Fre\'chet link is the second best model in the sense of recovery of parameters. Therefore, again, the skewed links are very flexible to fit a dataset that is considered imbalanced because they automatically estimate the degree of skewness by controlling the shape parameter to improve their performance.

In each simulation experiment, multiple chains are generated to produce Bayesian estimates. Based on both experiments, we deduce that the skewed link regression models outperform the symmetric link regressions when the binary outcome is highly imbalanced.

\begin{table}[!htbp]
  \centering
  \caption{Performance comparison of the various link models using simulation experiment I. The true model is based on the probit link with $(\beta_0, \beta_1, \beta_2) = (-3,  0.2, 0.7)$.}
  \scalebox{0.68}{
    \begin{tabular}{l|c|c|c|c|c|c}
    \multicolumn{1}{c}{}      & \multicolumn{1}{c}{logit} & \multicolumn{1}{c}{probit} & \multicolumn{1}{c}{cloglog} & \multicolumn{1}{c}{standard GEV} & \multicolumn{1}{c}{skewed Weibull} & \multicolumn{1}{c}{Fre\'chet} \\
    \hline
    Parameter & Estimate & Estimate & Estimate & Estimate & Estimate & Estimate \\
                   & 95\% CI (HPD) & 95\% CI (HPD) & 95\% CI (HPD) & 95\% CI (HPD) & 95\% CI (HPD) & 95\% CI (HPD) \\
    \hline
    $\beta_0$ & -5.8806 & -2.8146 & -5.8596 & -5.8609 & -4.6617 & -4.9927 \\
       & (-7.1661, -4.5951) & (-3.2984,-2.3309) & 
       (-7.1206,-4.5987) & 
       (-8.2872 ,-3.6667) & 
       (-7.2326,-2.0174) & 
       (-6.9264 -3.0682) \\
    $\beta_1$ & 1.169 & 0.4436 & 1.1694 & 1.0438 & 0.7564& 0.9623 \\
       & (-0.0638, 2.4018) & (-0.0687,0.956) & (-0.0347, 2.3736) & (-0.0746,2.4326) & (-0.2486, 1.8675) & (-0.3557, 2.0823) \\
    $\beta_2$ & 1.3656 & 0.5651 & 1.3237 & 1.2716 & 0.9187 & 0.8908 \\
       & (0.6954, 2.0359) & (0.2732, 0.8571) & (1.1147, 2.4244) & (0.6886, 1.9587) & (0.2459, 1.6592) & (0.383, 1.4744) \\
       &   &   &   &   &   &   \\
    shape parameter &       &       &      & 0.0032 & 1.2437 & 3.2626 \\
                             &       &       &       & (-0.1884, 0.1495) & (0.7824, 1.8776) & (2.3721, 4.406) \\
       &   &   &   &   &   &   \\
    \hline
       &   &   &   &   &   &   \\
    DIC   &       &       &      &  106.2756 & 103.5679 & 108.6163 \\
    -logLik &  49.8878 & 49.8387 & 49.9225 & 50.0309 & 51.3653 & 51.5532 \\
    BIC   &   120.4988 & 120.4007 & 120.5683 &   &   &  \\
    KS    & 0.9972 & 0.9976 & 0.9972 & 0.9961 & 0.9967 & 0.8433 \\
    MAE   & 0.0209 & 0.02095 & 0.02088 & 0.02186 & 0.02123 & 0.1808 \\
       &   &   &   &   &   &   \\
    \hline
      \end{tabular}}
  \label{tab:sim_probit}
\end{table}

\begin{table}[!htbp]
  \centering
  \caption{Performance comparison of the various link models using simulation experiment II. The true model is based on the standard GEV link with $(\beta_0, \beta_1, \beta_2) = (-3, -0.2, 0.8)$.}
  \scalebox{0.67}{
    \begin{tabular}{l|c|c|c|c|c|c}
    \multicolumn{1}{c}{}      & \multicolumn{1}{c}{logit} & \multicolumn{1}{c}{probit} & \multicolumn{1}{c}{cloglog} & \multicolumn{1}{c}{standard GEV} & \multicolumn{1}{c}{skewed Weibull} & \multicolumn{1}{c}{Fre\'chet} \\
    \hline
    Parameter & Estimate & Estimate & Estimate & Estimate & Estimate & Estimate \\
                   & 95\% CI (HPD) & 95\% CI (HPD) & 95\% CI (HPD) & 95\% CI (HPD) & 95\% CI (HPD) & 95\% CI (HPD) \\
    \hline
    $\beta_0$ &-6.9426 & -3.5930 & -7.2263 & -2.5713 & -2.9694 & -3.2936 \\
       & (-8.6051, -5.2801) & (-4.4801, -2.7059) & (-9.1241, -5.3285) & (-3.5931, -1.3858) & (-3.9025, -2.0236) & (-4.2359, -2.2512) \\
    $\beta_1$ & -0.6864 & -0.2800 & -0.5363 & -0.2301 & -0.2835 & -0.5350 \\
        & (-2.2951, 0.9223) & (-1.0877, 0.5277) & (-2.2800, 1.2074) & (-0.6807, 0.2135) & (-0.8473, 0.1638) & (-1.3672, 0.0370) \\
    $\beta_2$ & 2.4534 & 1.2739 & 2.6662 &  0.6581 & 0.5784 & 0.8103 \\
        & (1.6712, 3.2357)  & (0.7798, 1.7679) & (1.7058, 3.6265) & (0.1729, 1.1070) & (0.3802, 1.1798) & (0.4761, 1.2215) \\
       &   &   &   &   &   &   \\
    shape parameter &       &       &      &  -0.4523 & 1.8945 & 4.5823 \\
                             &       &       &       & (-0.7527, -0.2207) & (1.2508, 2.6177) & (3.3014, 6.1055) \\
       &   &   &   &   &   &   \\
    \hline
       &   &   &   &   &   &   \\
    DIC   &       &       &      &  81.8384 & 82.1210 & 86.8719 \\
    -logLik & 56.3239  & 37.0525 & 37.3855 & 37.2829 & 39.6527 & 58.1976 \\
    BIC   &   133.3711  & 94.8283 & 95.4942 &    &   &  \\
    KS    & 0.9991  & 0.9925 & 0.9918 & 0.9870 & 0.9869 & 0.7983 \\
    MAE   & 0.0264 & 0.0191 & 0.0187 & 0.0211 & 0.0204 & 0.1481 \\
       &   &   &   &   &   &   \\
    \hline
      \end{tabular}}
  \label{tab:sim_gev}
\end{table}

\section{Empirical data on life insurance} \label{sec:empirical}
\markboth{Empirical data on life insurance}{}

To illustrate the applicability of the binary regression models previously discussed, we drew observations of policyholders during calendar year 2015 from an insurance portfolio of a major insurer. In order to preserve confidentiality, we took subsamples of 127,777 observations for training our models and 18,254 observations for testing the various estimated models. During the calendar year, we are able to observe whether the policyholder died or survived during the year. The observed mortality is highly imbalanced with 1,720 actual deaths observed in our training set and 247 actual deaths observed in our test set. Each observation is described by 9 attributes of which 6 are considered categorical and 3 are numerical variables.

\subsection{Data description}

Table \ref{tab:data summary_part1} shows the description and summary statistics of each categorical variable used in our mortality investigation. We regroup the variable \textit{State} into 8 different geographic regions and create a new categorical variable \textit{Regions} in order to simplify the number of levels to describe geographic location. For many life insurance contracts, secondary guarantees are offered to ensure some death benefits are paid even if the cash surrender value of the policy falls to zero. From our portfolio, we observe six levels to describe the kind of secondary guarantees that is purchased with the contract. The insured's sex indicator \textit{Gender} is also a discrete variable with 2 levels, Male and Female, with 52.75\% Male and the rest are Female. \textit{Smoker Status} indicates the insured's smoking status at policy issue with $6.52 \%$ Smokers, $75.19 \%$ Non-Smokers, and $18.29 \%$ Unismokers. It is not uncommon practice to have simplified underwriting classifying a policyholder as Unismoker, in the absence of information about smoking habit and for which a blended smoker and non-smoker mortality rates are used. The Line of business (LOB) variable broadly classifies the type of life insurance business the contract is issued. Finally, the categorical variable \textit{Plan} has 7 different levels: Universal Life with Long Term Care (UL\_LTC), Term Insurance Plan (Term), Universal Life with Secondary Guarantees (UL\_SG) and without Secondary Guarantees (UL\_NSG),  Variable Life with Secondary Guarantees (VUL), Whole Life Policy and Corporate-Owned Life insurance (COLI). This categorization can be quite unique to the type of contracts issued by an insurance company.

Summary statistics of the numerical variables are in Table \ref{tab:data summary_part2}. The issue ages range from as young as newborn to as old as 90, with an average age at issue of 37 years. The face amount is the amount of death benefit and although this amount may vary according to the type of contract, most of our policies in our dataset have level face amounts. The duration, measured in years, refers to the number of years since policy issue, as of the period of observation. All three continuous variables can have direct effect on mortality, for example, it is well known that mortality rate increases with age.

In this paper,  we normalize these three continuous variables: Issue Age, Face Amount and Duration. For Issue Age and Duration, we rescale the variables so that the range is in $[-1,1]$, for which the general formula is given by: 
$$
x_{new} = \frac{x-\text{mean}(x)}{\text{standard deviation}(x)}, 
$$
where $x$ is the original value and $x_{new}$ is the normalized value. However, for \text{Face Amount}, the variable is highly skewed understandably because of some policies with unusually high values; its magnitude has to be narrowed down. In this case, we take logarithm of the total data followed by rescaling the data:
$$
x_{new} = \frac{\log(x)-\text{min}(\log(x))}{\text{max}(\log(x))-\text{min}(\log(x))}.
$$

\begin{table}[htbp]
  \centering
  \caption{Mortality data description and summary of categorical variables}
  \scalebox{0.9}{
    \begin{tabular}{rrlc}
    \multicolumn{1}{l}{Categorical variables } & \multicolumn{1}{l}{Description} &       & Proportions \\
    \midrule
    \multicolumn{1}{l}{Regions} & \multicolumn{1}{l}{Region of issue } & Northest\_NewEngland & 5.57\% \\
          &       & Northest\_MidAtlantic & 14.36 \\
          &       & Midwest & 23.13\% \\
          &       & South\_Atlantic & 24.45\% \\
          &       & South\_SouthCentral & 14.19\% \\
          &       & West\_Mountain & 4.92\% \\
          &       & West\_Pacific & 16.17\% \\
          &       & Foreign & 0.70\% \\
    \midrule
    \multicolumn{1}{l}{SG\_IND} & \multicolumn{1}{l}{Secondary guarantee indicator} & Brokerage & 0.45\% \\
          &       & Legacy  & 1.03\% \\
          &       & UL\_LTC & 6.97\% \\
          &       & NSG   & 31.28\% \\
          &       & No Information & 52.11\% \\
          &       & SG    & 3.16\% \\
    \midrule
    \multicolumn{1}{l}{Gender} & \multicolumn{1}{l}{Insured's sex} & Female & 42.75\% \\
          &       & Male  & 57.25\% \\
    \midrule
    \multicolumn{1}{l}{Smoker Status} & \multicolumn{1}{l}{Insured's smoking status} & Smoker & 6.52\% \\
          &       & Nonsmoker & 75.19\% \\
          &       & Unismoker & 18.29\% \\
    \midrule
    \multicolumn{1}{l}{LOB} & \multicolumn{1}{l}{Line of business} & ISL   & 20.75\% \\
          &       & TRAD  & 20.09\% \\
          &       & EM  & 1.04\% \\
          &       & No Information & 52.11\% \\
    \midrule
    \multicolumn{1}{l}{Plan} & \multicolumn{1}{l}{Plan type} & UL\_LTC & 8.44\% \\
          &       & Term  & 29.07\% \\
          &       & UL\_NSG & 31.83\% \\
          &       & UL\_SG & 1.05\% \\
          &       & VUL   & 5.89\% \\
          &       & Whole Life & 22.68\% \\
          &       & COLI  & 1.04\% \\
    \bottomrule
    \end{tabular}
    }
  \label{tab:data summary_part1}%
\end{table}%

\begin{table}[htbp]
  \centering
  \caption{Mortality data description and summary of numerical variables}
   \scalebox{0.92}{
    \begin{tabular}{lrrrrrr}
    Numerical variables  & Minimum & 1st Quantile & Mean & Median & 3rd Quantile  & Maximum \\
    \midrule
    Issue Age & 0  &  26 & 37 & 39   & 51    & 90 \\
    Face Amount (in \$) & 0   & 25,000 & 314,902 & 100,000 & 250,000 & 12,000,000 \\
    Duration (in years) & 0.3 & 2.5 &5.6  & 4.5 & 7.5 & 23.8 \\
    \bottomrule
    \end{tabular}
    }
  \label{tab:data summary_part2}
\end{table}

\subsection{Numerical results}

We calibrated the binary regression models under the six different link functions as earlier described. We estimated the parameters using the training dataset; the test dataset has been used for validation. We present the results separately for the standard link functions (logit, probit, cloglog) and the skewed link functions (standard GEV, skewed Weibull, and Fre\'chet). We will call the standard link models as frequentist models while the skewed link models as Bayesian models. This differentiation has to do with reference to the approach used to estimate the parameters. Note that for the categorical variables, the reference level should be clear from the results presented. For example, in the case of Plan type, the category COLI has been used; in the case of Line of business (LOB), the category EM has been used. 

Table \ref{tab:reduced_freq} provides the parameter estimates for the reduced frequentist models. First, we fit a full model with all the predicted variables included in the model, then the reduced model is deployed with insignificant variables excluded in the model. Interestingly, in terms of model comparison statistics, the probit model has the worst log likelihood, KS, and MAE, but it has the best BIC statistic. As the probit model rejects many predictor variables, its low BIC appears to intuitively make sense. The probit model would be appealing to those seeking for the simplest possible model, and yet still considered to be competitive. The probit model chooses only IssueAge and Duration for significant variables, and the parameter estimates for this model intuitively make sense. The probability of death increases with IssueAge and also slightly increases with policy duration.

\begin{table}[!htbp]
  \centering
  \caption{Parameter estimates for the reduced frequentist models}
  \scalebox{0.9}{
    \begin{tabular}{l|cc|cc|cc}
     \multicolumn{1}{c}{} & \multicolumn{2}{c}{logit} & \multicolumn{2}{c}{probit} & \multicolumn{2}{c}{cloglog} \\
    \hline
    Parameter & Estimate  & (s.e.) & Estimate  & (s.e.) & Estimate  & (s.e.) \\
    \hline
    Intercept &  -8.1418 & (0.8751) & -4.8315 & (0.1630) & -10.8632 & (0.5227)\\
    GenderMale & 0.4493 & (0.1473)  &       &       & 0.4922 & (0.1434)\\
    IssueAge & 7.5972 & (0.5803) & 3.2135 & (0.2098) & 8.1890 & (0.5719) \\
    SmokerStatusS & 0.6936 & (0.1940) &       &       &  0.5500 & (0.1834) \\
    SmokerStatusUni & -0.1080 & (0.2607) &       &       & -0.2090 & (0.2545) \\
    Duration & 0.0767 & (0.0074) & 0.0423 & (0.0022) & 0.0937 & (0.0066) \\
    FaceAmount & -3.0868 & (1.0348)  &       &       &       &       \\
    Plan\_typeUL\_LTC & -0.3088 & (0.3109)  &       &       & -0.4393 & (0.3144)\\
    Plan\_typeTERM & -0.8083 & (0.3736) &       &       & -1.1180 & (0.4024) \\
    Plan\_typeUL\_NSG & -0.1871 & (0.2199) &       &       & -0.0905 & (0.2077) \\
    Plan\_typeUL\_SG & -1.4596 & (1.0371) &       &       & -0.0670 & (0.4984) \\
    Plan\_typeVUL & -0.7097 & (0.5478)  &       &       & -0.3484 & (0.4577) \\
          &       &       &       &       &       &  \\
    -log Lik & \multicolumn{2}{c|}{1044.8173} & \multicolumn{2}{c|}{1068.5629} & \multicolumn{2}{c}{1039.393} \\
    BIC   & \multicolumn{2}{c|}{2207.3802} & \multicolumn{2}{c|}{2166.5623} & \multicolumn{2}{c}{2186.719} \\
    KS    & \multicolumn{2}{c|}{0.9994} & \multicolumn{2}{c|}{0.9999} & \multicolumn{2}{c}{0.9992} \\
    MAE   & \multicolumn{2}{c|}{0.0251} & \multicolumn{2}{c|}{0.0254} & \multicolumn{2}{c}{0.0253} \\
    \hline
    \end{tabular}}
  \label{tab:reduced_freq}
\end{table}

Table \ref{tab:red_bayes} presents the parameter estimates for the reduced Bayesian models. As stated before, one statistic used to compare Bayesian models is the DIC, which describes the trade-off between quality of the model fit to the data and the complexity of the model. The model that gives the smaller DIC is generally preferred as it is better supported by the observed data. Among all three skewed link models, the skewed Weibull gives the best DIC value of 2126.3532. While this model has the worst KS statistic, it still outperforms the other two models in terms of loglikelihood and MAE values.
\begin{table}[!htbp]
  \centering
  \caption{Parameter estimates for the reduced Bayesian models}
  \scalebox{0.9}{
    \begin{tabular}{l|c|c|c}
      \multicolumn{1}{c}{}  & \multicolumn{1}{c}{Standard GEV} & \multicolumn{1}{c}{Skewed Weibull} & \multicolumn{1}{c}{Fre\'chet} \\
    \hline
    Parameter & Estimate & Estimate & Estimate \\
      & 95\% HPD & 95\% HPD & 95\% HPD \\
    \hline
    Intercept & -5.3912 & -5.2117 & -4.7430 \\
     & (-7.9279, -3.9688) & (-7.6501, -3.8083) & (-8.5038, -3.2466) \\
    IssueAge & 2.0595 & 4.3174 & 1.8652 \\
     & (1.0426,  3.4020) & (2.8483,  6.9725) & (0.9955,  3.2125) \\
    SmokerStatusS & 0.4927 & 0.3434 & 0.3651 \\
     & (0.1658, 0.9199) & (0.0792, 0.6566) & (0.0749, 0.7524) \\
    SmokerStatusUni & 0.2028 & -0.0226 & 0.0955 \\
     & (-0.1612, 0.6155) & (-0.3849, 0.2493) & (-0.2651, 0.4372) \\
    Duration &       &       & 0.9030  \\
     &       &       & (0.5776,  1.4798) \\
    LOBISL & 2.2296 &   & 1.9686 \\
     & (0.8777, 4.0847) &    & (0.6643, 4.2602) \\
    LOBNoInformation & 2.0458 &   & 1.6797 \\
     & (0.7201, 3.6112) &    & (0.6770, 4.4102) \\
    LOBTRAD & 1.6730 &   & 1.4174 \\
     & (0.3112, 3.2024) &    & (0.1673, 3.8592) \\
    Plan type UL\_LTC &  -1.5931  & -0.2598 &  -1.6166 \\
     & (-2.5892, -0.8178) & (-0.6591, 0.0973) & (-2.6182, -0.8015) \\
    Plan type TERM & -1.8061 & -0.5262 & -2.8171 \\
     & (-3.1911, -0.9376) & (-1.1886, -0.0119) & (-5.0451, -1.2020) \\
    Plan type UL NSG & -1.2067 & -0.2039 & -1.2311 \\
     & (-1.9193, -0.6662) & (-0.4883, 0.0602) & (-1.9962, -0.6076) \\
    Plan type UL SG & -2.1475 & -1.4483 & -2.4779 \\
     & (-3.7416, -0.9047) & (-3.6883, -0.0552) & (-4.4446, -1.0303) \\
    Plan type VUL & -1.8975 & -0.5691 & -2.4099 \\
     & (-3.2817, -0.9542) & (-1.4931, 0.0668) & (-4.1592, -1.0580) \\
        &       &       &       \\
    shape parameter & -0.1262 & 1.3185 & 3.8392 \\
     & (-0.2498, 0.0266) & (1.0112, 1.6081) & (2.7269, 5.0725) \\
        &       &       &       \\
    DIC   & 2338.5129 & 2126.3532 & 2348.6408 \\
    -log Lik & 1153.3895 & 1055.0456 & 1165.1385 \\
    KS    & 0.9990 & 0.9997 & 0.9980 \\
    MAE   & 0.0263 & 0.0251 & 0.0262 \\
    \hline
    \end{tabular}}
  \label{tab:red_bayes}
\end{table}

Interestingly, as in the frequentist models, the skewed Weibull link model also is the simplest model among all three skewed link models as it rejected more predictors than any other. In particular, both policy duration and LOB are considered not significant predictors under the skewed Weibull link model. We can roughly interpret the estimated coefficients of the selected predictor variables for the skewed link model, and it appears to be quite intuitive. First, the positive sign for the IssueAge indicates a greater likelihood of death for older issue ages. Second, the positive sign for SmokerStatusS indicates a worse mortality for smokers than for non-smokers. The negative sign for SmokerStatusUni might be counterintuitive, however, this is more challenging to interpret as this smoking status is about the absence of smoking habits. Finally, all the signs for the Plan types are negative, which indicates that Plan type COLI has the best mortality among all plan types. Because it has the largest negative sign, Plan type UL SG has the worse mortality among all types. Notice further that for all Bayesian models, the FaceAmount was not considered a significant predictor of death. There are two possible explanations to this. First, the effect is depleted when we transform and normalize the variable. Second, the amount of death benefit can be directly linked to the degree of underwriting adopted prior to issue. In which case, the mortality selection effect can be explained by the policy duration.

If we wish to compare the frequentist models to the Bayesian models, we could examine the common model statistics such as the loglikehood, KS, and MAE values. For all models, the KS and MAE statistics are relatively comparable. However, in terms of loglikelihood, the Bayesian models outperform all the frequentist models. In particular, the skewed Weibull link model has the smallest value for all of them.

\begin{table}[htbp]
  \centering
  \caption{Aggregated amount of death benefits and death counts for the various link models}
    \begin{tabular}{l|rr|lr}
    \multicolumn{1}{c}{} & \multicolumn{2}{c}{Amount of death benefits} & \multicolumn{2}{c}{Death counts} \\
    \hline
           &    &    &    & \\
    Actual amount & \$\,33,992,610 &       & Actual count & 247 \\
           &    &    &    & \\
          & Estimated & Error rate &       & Estimated \\
     logit & \$\,28,653,690 & 15.7\% &  logit & 19 \\
     probit & \$\,42,883,180 & 26.2\% &  probit & 17 \\
     cloglog & \$\,37,309,376 & 9.8\% &  cloglog & 19 \\
          &    &    &    & \\
     Standard GEV & \$\,43,623,689 & 28.3\% &  Standard GEV & 252 \\
     Skewed Weibull & \$\,36,338,219 & 6.9\% &  Skewed Weibull & 248 \\
     Fre\'chet & \$\,39,838,452 & 17.2\% & Fre\'chet & 250 \\
     \hline
    \end{tabular}
  \label{tab:agg_benefit}
\end{table}

\begin{figure}[!htbp]
\centering
   \includegraphics[width=0.8\textwidth]{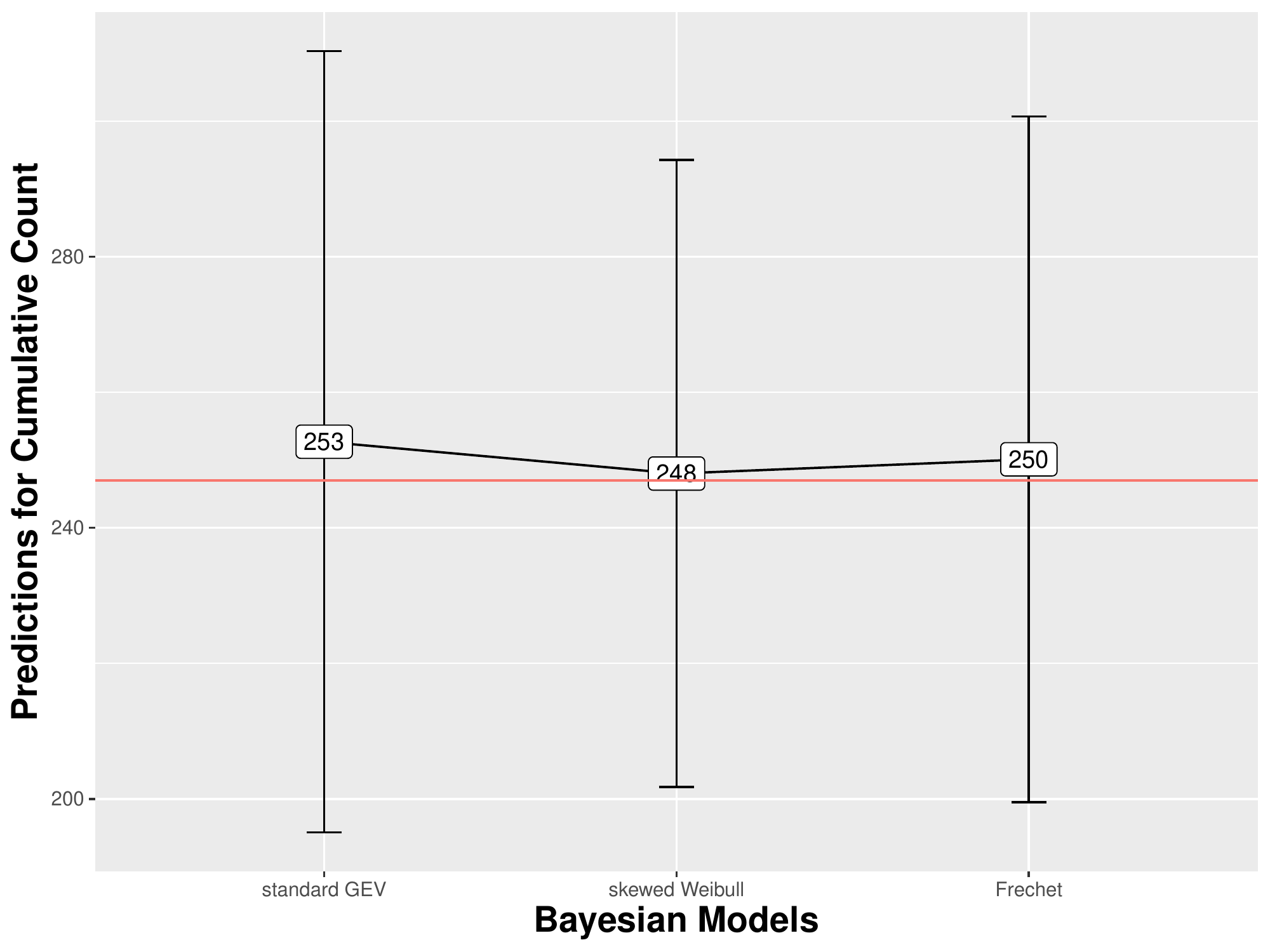}
   \includegraphics[width=0.8\textwidth]{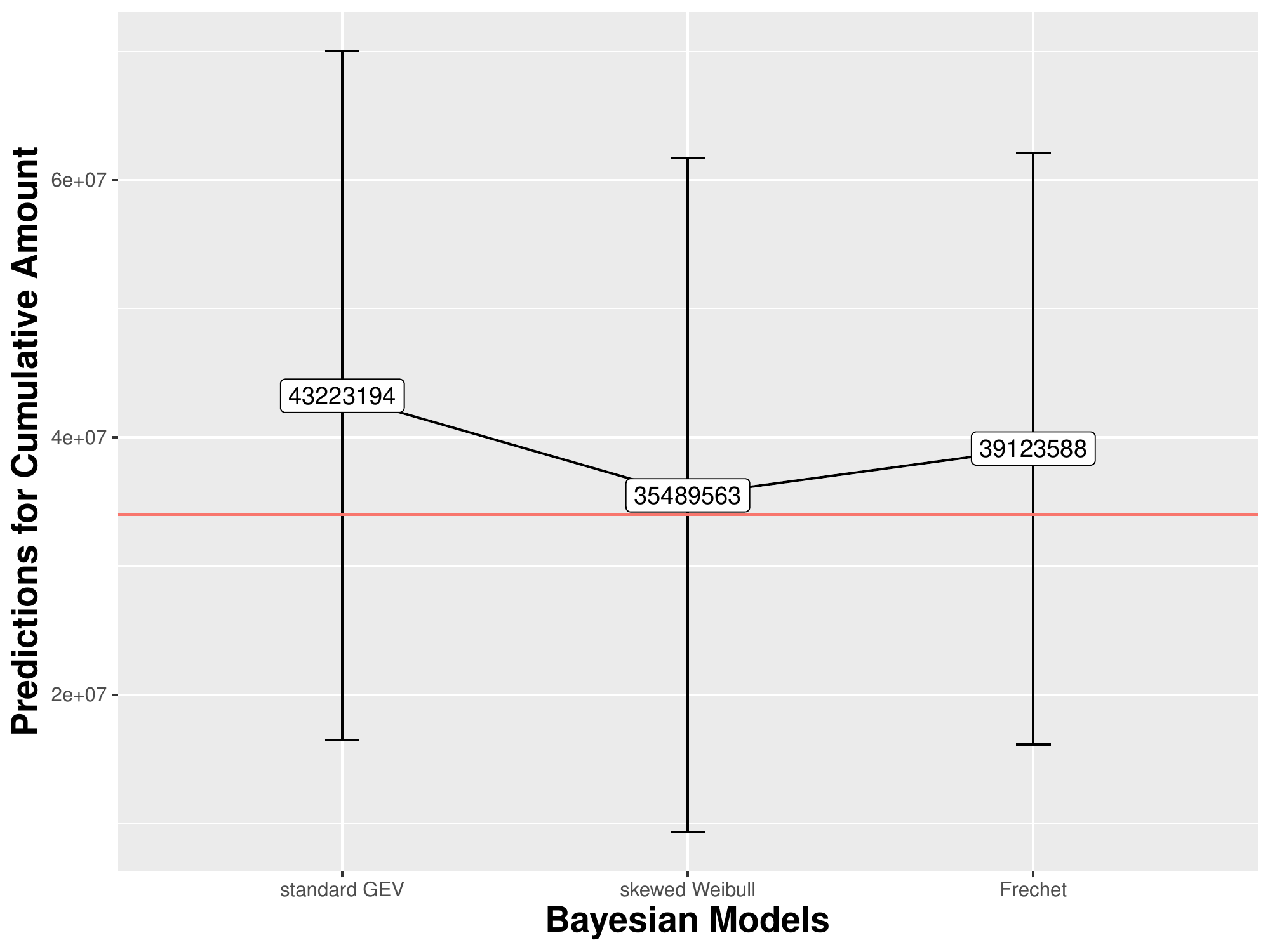}
  \caption{Mean and confidence intervals of the predictions from the Bayesian models for cumulative count and amount}
  \label{fig:pred}
\end{figure}

To further investigate the predictive power of the various models, we compare the prediction of death counts together with the corresponding death benefit amount. Table \ref{tab:agg_benefit} presents the aggregated death benefits and counts for the various link models using the test dataset. Accordingly, our test dataset has a total death count of 247, out of 18,254 observations; the corresponding benefit amount sums to approximately 34 million dollars. In terms of death counts, as we expected, the frequentist models could hardly come close to the true actual death counts; all three models underestimate the number of deaths. On the other hand, the skewed link models have predictions that are generally close to the actual death counts as shown in the table. In particular, the skewed Weibull, not surprisingly, marginally outperforms the other two models in terms of death counts, but is far superior in terms of predicting the corresponding death benefits. A visualization of the superiority of the skewed Weibull link model in terms of prediction is presented in Figure \ref{fig:pred}.

\section{Concluding remarks} \label{sec:conclude}
\markboth{Concluding remarks}{}

This article was driven by a purpose of developing a mortality investigation for tracking and monitoring death claims experience of a portfolio of life insurance contracts over a specified period. In this case, our primary interest is in developing a suitable binary regression model for predicting mortality, given a set of available covariate information. Our period of investigation is primarily short term, in this case, one year. Such is typical in practice where insurers need to predict the level of mortality for the following year given characteristics of their portfolio at the beginning of the year. However, this is particularly a challenging task as the observed mortality for a cohort of policies is expected to be a rare event.

This paper presents a methodology to handle extremely imbalanced binary outcome. In particular, we focus on the latent variable interpretation of the binary regression model and examine alternative functions that link the mean of the binary outcome to the predictor variables. We find that suitable link functions are those that have the flexibility to handle skewness, and we find three possible skewed link models: standard GEV, skewed Weibull, and Fre\'chet. We performed simulation studies and calibrated models using empirical data observed from a portfolio of life insurance contracts. Based on the model estimates and predictions, we find that all three models generally outperform frequentist models with standard link functions that include logit, probit, and cloglog. We also presented a Bayesian approach to the estimation of the skewed link models. As pointed out in this paper, the latent variable interpretation connects the various link models but it also has computational advantages especially when Bayesian inference is used for estimation such as the use of Metropolis-Hastings algorithm, an MCMC method.

The models and theoretical framework presented here can be extended to mortality tracking and monitoring on a more frequent basis than annually, for instance quarterly. This is because insurance companies publish quarterly financials to rating analysts and death benefits are a significant portion of an insurer's financials. The ability to understand and project death benefits on a quarterly basis is important to explain fluctuations in financials driven by mortality claims. However, when doing more frequent mortality tracking and monitoring than annual, it is important to distinguish between non-significant fluctuations in mortality and mortality fluctuations that represent a longer term trend. It would be interesting to explore this issue using the models we have developed in this paper. In addition, there is also the promise of using neural networks to handle imbalanced binary outcomes, and such would be an interesting future direction of research. See \cite{dreiseitl2002ann} and \cite{wang2016dnn}.

\newpage

\section*{Appendix} \label{appendix}
\markboth{Appendix}{}

The following useful results provide detailed proof that in the case of the \textit{Fre\'chet} link model that the resulting posterior distribution is proper. 
\begin{prop} \label{prop1}
Let $\boldsymbol{x}$ be the design matrix with full rank and $\boldsymbol{x}^{*}$ be the matrix with the $i$th row of $\epsilon_i \boldsymbol{x}_i'$, where $\epsilon_i= I(y_i=1) - I(y_i=0)$, for $i=1, 2, \ldots, n$. Then there exists a positive vector $\boldsymbol{a}=(a_1, a_2, \ldots, a_n)'$ with $a_i>0$ such that $\boldsymbol{a}' \boldsymbol{x}^{*} =0$. Under the prior given by $\pi(\alpha)\propto 1/\alpha^c$ with $c>1$, the posterior distribution stated in (\ref{eq:jointFR}) is proper. 
\end{prop}

\begin{proof}
Let $u_1, u_2, \ldots, u_n$ be independent random variables with common Fre\'chet distribution with $\mu=0, \sigma = 1$, and shape parameter $\alpha$. For $0 < k < 1$, it can be shown that $\E(|u|^k) = \Gamma(1-\frac{k}{\alpha}) < \infty$ for $\alpha > 1$. Observing that  $1-F(-x)=\E(I(u>-x))$ and $F(-x)=\E(I(-u \leq x))$, where $I(\cdot)$ is the indicator function. Then

\begin{equation}
\left ( 1-F(-\boldsymbol{x_i'\beta})\right)^{y_i} \left ( F(-\boldsymbol{x_i'\beta})\right)^{1-y_i} \leq \E \left [I \left \{ \epsilon_i u_i \geq \epsilon_i (-\boldsymbol{x_i'\beta}))\right  \}\right ]
\end{equation}
and
\begin{equation}
\left ( 1-F(-\boldsymbol{x_i'\beta})\right)^{y_i} \left ( F(-\boldsymbol{x_i'\beta})\right)^{1-y_i} \geq \E \left [I \left \{ \epsilon_i u_i > \epsilon_i (-\boldsymbol{x_i'\beta}))\right  \}\right ]
\end{equation}
Let $\boldsymbol{u}^* = (\epsilon_1 u_1, \cdots, \epsilon_n u_n)$. Using Fubini's Theorem, we obtain 
\begin{equation}
\begin{aligned}
& \int_{1}^{\infty} \int_{R_k} L(\boldsymbol{\beta}, \alpha|\boldsymbol{y}, \boldsymbol{x}) \frac{1}{\alpha^c} d\boldsymbol{\beta} d\alpha \\
 = & \int_{1}^{\infty} \frac{1}{\alpha^c} \int_{R_n}  \E \left[ \int_{R_k} I(-\epsilon_i \boldsymbol{x_i'\beta} < \epsilon_i u_i, 1\leq i \leq n) d \boldsymbol{\beta}  \right] dF(\boldsymbol{u}) d\alpha \\
 =& \int_{1}^{\infty} \frac{1}{\alpha^c} \int_{R_n}\E \left[\int_{R_k} I(\boldsymbol{x}^* \boldsymbol{\beta} <  \boldsymbol{u}^* ) d \boldsymbol{\beta}  \right]dF(\boldsymbol{u}) d\alpha 
\end{aligned}
\end{equation}
Directly following from Lemma 4.1 of \cite{chen2001propriety}, and under the condition that $\E(|u|^k) < \infty$ and $\displaystyle \int_1^{\infty}\frac{1}{\alpha^c} d \alpha < \infty$ for $c>1$, there exists a constant $K$ such that
\begin{equation}
\int_{R_k} I(\boldsymbol{x}^* \boldsymbol{\beta} <  \boldsymbol{u}^* ) d \boldsymbol{\beta} \leq K \left \Vert \boldsymbol{u}^* \right \Vert^k 
\Longrightarrow  \int_{1}^{\infty} \int_{R_k} L(\boldsymbol{\beta}, \alpha|\boldsymbol{y}, \boldsymbol{x}) \frac{1}{\alpha^c} d\boldsymbol{\beta} d\alpha < \infty.
\end{equation}
\end{proof}

\newpage

\bibliographystyle{apalike}
\bibliography{binimb}

\end{document}